\documentclass[journal,twocolumn]{IEEEtran}

\usepackage{epsfig}
\usepackage{subfigure}
\usepackage{graphicx}

\usepackage{amsmath}
\usepackage{amssymb}
\usepackage{color}
\usepackage{cite}

\def\BibTeX{{\rm B\kern-.05em{\sc i\kern-.025em b}\kern-.08em
    T\kern-.1667em\lower.7ex\hbox{E}\kern-.125emX}}

\newtheorem{claim}{Claim}

\newtheorem{theorem}{Theorem}
\newtheorem{definition}{Definition}
\newtheorem{lemma}{Lemma}

\newtheorem{remark}{Remark}

\title{Computation in Multicast Networks: \\ Function Alignment and Converse Theorems}

\author{Changho Suh, Naveen Goela and Michael Gastpar \\
\thanks{This work was presented in part at the IEEE Information Theory Workshop 2012, and the Allerton Conference 2012.}
\thanks{C.~Suh is with the School of Electrical Engineering, KAIST, South Korea (Email: {\sf chsuh@kaist.ac.kr}).}
\thanks{N.~Goela is with Technicolor Research, Los Altos, CA, USA (Email: {\sf naveen.goela@gmail.com}).}
\thanks{M.~Gastpar is with the School of Computer and Communication Sciences, EPFL, Switzerland (Email: {\sf michael.gastpar@epfl.ch}), and with the Department of Electrical Engineering and Computer Sciences, University of California, Berkeley, Berkeley, USA (Email: {\sf gastpar@eecs.berkeley.edu}).}
\thanks{C.~Suh was supported by Institute for Information \& communications Technology Promotion (IITP) grant funded by the Korea government (MSIP) (No.R0190-15-2030, Reliable crypto-system standards and core technology development for secure quantum key distribution network).}
}

\begin{document}

\maketitle

\begin{abstract}
The classical problem in network coding theory considers \emph{communication} over multicast networks. Multiple transmitters send independent messages to multiple receivers which decode the same set of messages. In this work, \emph{computation} over multicast networks is considered: each receiver decodes an identical \emph{function} of the original messages. For a countably infinite class of two-transmitter two-receiver single-hop linear deterministic networks, the computation capacity is characterized for a linear function (modulo-$2$ sum) of Bernoulli sources. A new upper bound is derived that is tighter than cut-set based and genie-aided bounds. A matching inner bound is established via the development of a \emph{network decomposition theorem} which identifies elementary parallel subnetworks that can constitute an original network without loss of optimality. The decomposition theorem provides a conceptually-simple proof of achievability that generalizes to $L$-transmitter $L$-receiver networks.
\end{abstract}

\section{Introduction}

Recently coding for computation in networks has received considerable attention with applications in sensor networks~\cite{PRKumar:it} and cloud computing scenarios~\cite{DimakisKannan:it10,Dimakis:Proceedings}. In a sensor network, a fusion node may be interested in computing a relevant function of the measurements from various data nodes. In a cloud computing scenario, a client may download a function or part of the original source information that is distributed (e.g. using a maximum distance separable code) across multiple data nodes.

The simplest setting for computation in networks consists of multiple sources transmitting information to a \emph{single} receiver which computes a function of the original sources. Appuswamy \emph{et al.} study the fundamental limits of computation for linear and general target function classes for single-receiver networks~\cite{Appuswamy:it11}. While limited progress has been made for general target functions, the problem of linear function computation in single-receiver networks has been solved in part due to a duality theorem establishing an equivalence to the classical problem of communication over multicast networks~\cite{ahlswede:it}.

Several results over the past decade have contributed to the understanding of classical communication in multicast networks in which the task is to transmit raw messages from transmitters to a set of receivers with identical message demands. The celebrated work of Ahlswede \emph{et al.}~\cite{ahlswede:it} established that the cut-set bound is tight for multicast communication. Subsequent research developed practical linear network coding strategies ranging from random linear codes to deterministic polynomial-time code constructions~\cite{Yeung:it,Koetter:it,Ho:it,Jaggi:it}. The success of traditional multicast communication motivates us to explore the fundamental limits of multicasting a linear function in \emph{multiple-receiver} networks as a natural next step. For this open problem, some facts are known based on example networks: (a) Random codes are insufficient in achieving capacity limits, and structured codes achieve higher computation rates~\cite{NazerGastpar:it07}; (b) Linear codes are insufficient in general for computation over multi-receiver networks (cf. both~\cite{RaiDey:it12} and~\cite{Dougherty:it05}) and non-linear codes may achieve higher computation rates.

To make progress on the problem of multicasting a function, we consider a multiple-receiver network scenario in which all of the receivers compute a linear function (modulo-$2$ sum) of the independent Bernoulli sources generated at transmitters. Specifically, we consider the Avestimehr-Diggavi-Tse (ADT) deterministic network model~\cite{Salman:IT11} which well abstracts wireless Gaussian networks. In the context of classical communication, it has been well known that ADT networks can approximate wireless Gaussian networks within a constant gap to the optimality in capacity~\cite{Salman:IT11,BT:Euro}. Recently a similar approximation result has been established for the problem of computation in which a single receiver wishes to compute a linear function of multiple Gaussian sources~\cite{ZhanPark:JSAC}. Specifically \cite{ZhanPark:JSAC} employs lattice codes to show that a multiple-source single-destination Gaussian network can be approximated to a class of linear deterministic networks (which includes the ADT network as a special case), within a constant factor of the optimal performance w.r.t. the distortion for computing the sum of the Gaussian sources. In this work, we intend to extend this approximation approach to more general computation scenarios in Gaussian networks. As an initial effort, we consider an $L$-transmitter $L$-receiver ADT network with the function multicast demand.

In this paper, we derive a new upper bound that is tighter than cut-set based bounds and genie-aided bounds. Especially in the case of $L=2$, we establish a matching inner bound to characterize the computation capacity. The achievability builds upon our development of a \emph{network decomposition theorem} which identifies elementary parallel subnetworks that can constitute an original network without loss of optimality. The network decomposition offers a conceptually simple achievability proof which we use to generalize to an arbitrary value of $L$. In the $L$-user case, we show the optimality of our achievability in the limit of $L$. Our achievable scheme is intimately related to the concept of \emph{interference alignment} although the purpose of alignment is different. In our problem, the alignment idea is employed to compute a desired function with a smaller number of linear-subspace signal dimension than the number of sources involved in the function.

{\bf Related Work:} In~\cite{Ramamoorthy:ISIT,RaiDey:it12,RamamoorthyLangberg:journal}, the computation capacity for multicasting a sum of sources is explored for arbitrary multiple-source multiple-destination networks. Rai and Dey~\cite{RaiDey:it12} proved that there exists a linear solvably equivalent sum-network for any multiple-unicast network and vice-versa. Ramamoorthy and Langberg~\cite{RamamoorthyLangberg:journal} characterized necessary and sufficient conditions for communicating sums of sources of two-source $L$-destination (or $L$-source two-destination) networks, when the entropy of each source is limited by $1$. On the other hand, we consider sources without entropy constraints and establish the exact computation capacity of an ADT multiple-receiver network.

\section{Model}
\label{sec:model}

\begin{figure}[t]
\begin{center}
{\epsfig{figure=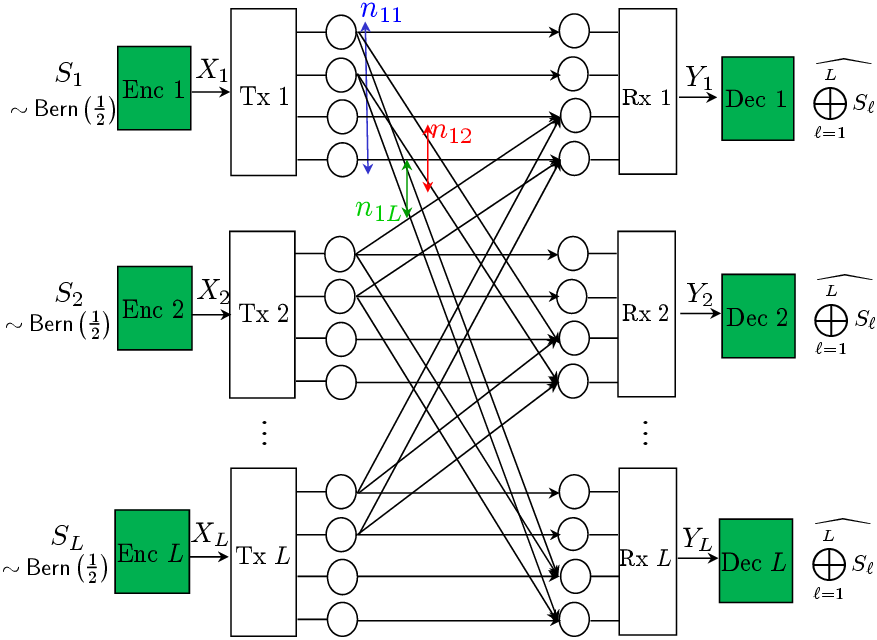, angle=0, width=0.45\textwidth}}
\end{center}
\caption{$L$-transmitter $L$-receiver Avestimehr-Diggavi-Tse (ADT) deterministic network.}
\label{fig:model}
\end{figure}

Consider an $L$-transmitter $L$-receiver ADT deterministic network depicted in Fig.~\ref{fig:model}. This network is described by integer parameters $n_{ij}$'s, each indicating the number of signal bit levels from transmitter $i$ $(i=1, \cdots, L)$ to receiver $j$ $(j=1,\cdots, L)$. In the example of Fig.~\ref{fig:model}, $n_{11}=4$, $n_{12}=2$ and $n_{1L}=2$.
Let $X_{\ell} \in \mathbb{F}_2^{q}$ be transmitter $\ell$'s encoded signal where $q=\max_{ij} n_{ij}$. The received signal at Rx $\ell$ is then given by
\begin{align}
\begin{split}
\label{eq:receivedsignals}
Y_{\ell} = {\bf G}^{q-n_{1 \ell}} X_1 \oplus {\bf G}^{q-n_{2 \ell }} X_2 \oplus \cdots \oplus {\bf G}^{q-n_{L \ell }} X_L,  \end{split}
\end{align}
where ${\bf G}$ is the $q$-by-$q$ shift matrix, i.e., $[{\bf G}]_{ij} = \mathbf{1} \{ i = j+ 1\}$ ($1 \leq i \leq q; 1 \leq j \leq q$), and operations are performed in $\mathbb{F}_2$. Here $\oplus$ indicates the bit wise XOR.

Each receiver wishes to compute modulo-2 sums of the $L$ Bernoulli sources $(S_1^K, \cdots, S_{L}^K)$, generated at the $L$ transmitters, with $N$ uses of the network. Here we use shorthand notation to indicate the sequence up to $K$, e.g., $S_{\ell}^K:=(S_{\ell 1}, \cdots, S_{ \ell K})$. We assume that  $(S_1^K, \cdots, S_{L}^K)$ are independent and identically distributed with ${\sf Bern} (\frac{1}{2})$. Transmitter $\ell$ uses its encoding function to map $S_\ell^K$ to a length-$N$ codeword $X_\ell^N$. Receiver $\ell$ uses a decoding function $d_\ell$ to estimate $\bigoplus_{\ell=1}^{L }S_{\ell}^K$ from its received signal $Y_\ell^N$.  An error occurs whenever $d_\ell \neq \bigoplus_{\ell=1}^{L }S_{\ell}^K$. The average probabilities of error are given by
$\lambda_{\ell} = \mathbb{E} [ \textrm{Pr} (d_\ell \neq \bigoplus_{\ell=1}^{L }S_{\ell}^K) ], \ell=1,\cdots,L$.

We say that the computation rate $R_{\sf comp}=\frac{K}{N}$ is achievable if there exists a family of codebooks and  encoder/decoder functions such that the average decoding error probabilities of $(\lambda_1, \cdots, \lambda_L)$ go to zero as code length $N$ tends to infinity. The computation capacity $C_{\sf comp}$ is the supremum of the achievable computation rates.

In the case of $L=2$, we classify networks into two classes depending on a channel parameter condition.
\begin{definition}
\label{def:networkclass}
A network is said to be \emph{degenerate} if $n_{11}-n_{12}=n_{21}-n_{22}$. A network is said to be \emph{non-degenerate} if $n_{11}-n_{12} \neq n_{21}-n_{22}$.
\end{definition}
\begin{remark}
Suppose that $n_{11}-n_{12}=n_{21}-n_{22} \geq 0$. In this case, one can see that
\begin{align*}
{\bf G}^{n_{21}-n_{22}} Y_1 &= {\bf G}^{q-n_{11}+ n_{21}-n_{22}} X_1 \oplus {\bf G}^{q-n_{22}} X_{2}  \\
& = {\bf G}^{q-n_{12}} X_1 \oplus {\bf G}^{q-n_{22}} X_{2} = Y_2.
\end{align*}
implying that $Y_2$ is a \emph{degenerated} version of $Y_1$. For the other case of $n_{11}-n_{12}=n_{21}-n_{22} \leq 0$, one can readily see that $Y_1$ is a degenerated version of $Y_2$ as $Y_1 = {\bf G}^{n_{12} - n_{11}} Y_2$. $\square$
\end{remark}

\section{Main Results}

\subsection{2-by-2 Network}

\begin{theorem}
\label{thm:computingcapacity}
\begin{align}
\label{eq:computingcapacity_deg}
C_{\sf comp} \leq \min \{ n_{11},n_{12},n_{22},n_{21} \}.
\end{align}
For degenerate networks: $n_{11} - n_{12} = n_{21}  - n_{22}$, this upper bound is achievable.
\end{theorem}
\begin{proof}
The standard cut-set argument establishes the upper bound. Using Fano's inequality: $H(S_1^K \oplus S_2^K  | Y_1^N ) \leq 1 + \textrm{Pr} \left( d_1 \neq S_1^K \oplus S_2^K \right)  K$, and denoting $\epsilon_N : = 1/N +\textrm{Pr} \left( d_1 \neq S_1^K \oplus S_2^K \right)K/N$, we get:
 \begin{align*}
 N(R_{\sf comp} -\epsilon_N) &\leq I(S_1^K \oplus S_2^K; Y_1^N) \\
 &\leq I(S_1^K \oplus S_2^K; Y_1^N| S_2^K, X_2^N) \\
 & \leq H(Y_{1}^N|X_2^N) \leq Nn_{11}
 \end{align*}
  where the second inequality follows from the non-negativity of mutual information and the fact that $S_2^K$ is independent of $S_1^K \oplus S_2^K$. This yields $R_{\sf comp} \leq n_{11}$. Similarly one can prove that $R_{\sf comp} \leq \min \{ n_{12}, n_{21}, n_{22} \}$.

The achievability proof is as follows. Assume that $n_{11}-n_{12}=n_{21}-n_{22} \geq 0$. Then $Y_2$ is a degenerated version of $Y_1$: $Y_2 = {\bf G}^{n_{21}-n_{22}} Y_1$. This shows an equivalence to a single-receiver case which concerns receiver 2's demand only. Hence, the computation rate in this case is the same as that of a multiple-access channel having receiver 2 as the receiver. So $R_{\sf comp} \geq \min \{ n_{12}, n_{22} \}$~\cite{NazerGastpar:it07,Appuswamy:it11}. Similarly for the other case of $n_{11}-n_{12}=n_{21}-n_{22} \leq 0$, one can show that $Y_1$ is a degenerated version of $Y_2$ and therefore the network becomes equivalent to a single-receiver network w.r.t. receiver 1 where $R_{\sf comp} \geq \min \{ n_{11}, n_{21} \}$. In the first case, $n_{12} \leq n_{11}$ and $n_{22} \leq n_{21}$; hence $\min \{n_{12}, n_{22} \} \leq \min \{ n_{11}, n_{21} \}$. In the second case, on the other hand, $\min \{n_{12}, n_{22} \} \geq \min \{ n_{11}, n_{21} \}$. Therefore, $R_{\sf comp} \geq \min \{ n_{12}, n_{22}, n_{11}, n_{21} \}$.
\end{proof}

\begin{theorem}[Upper Bound for Non-degenerate Networks]
\label{thm:upperbound}
\begin{align}
\begin{split}
\label{eq:upper2}
C_{\sf comp} \leq \frac{ \max(n_{11},n_{21}) + \max(n_{22},n_{12}) }{3}.
 \end{split}
\end{align}
\end{theorem}
\begin{proof}
See Section~\ref{sec:proofofupperbound}.
\end{proof}

We show the tightness of the above bound for the case of $n:=n_{11} = n_{22}$ and $m:=n_{12}=n_{21}$ that we call a symmetric case.

\begin{theorem}[Symmetric Network]
\label{thm:sym_computingcapacity}
For $n:=n_{11} = n_{22}$ and $m:=n_{12}=n_{21}$,
\begin{align}
\label{eq:sym_computingcapacity}
C_{\sf comp}= \left\{
  \begin{array}{ll}
    \min \left \{ m,n, \frac{2}{3} \max(m,n) \right \}, & \hbox{$m \neq n$;} \\
    n, & \hbox{$m=n$.}
  \end{array}
\right.
\end{align}
\end{theorem}
\begin{proof}
The converse proof is immediate from Theorems~\ref{thm:computingcapacity} and~\ref{thm:upperbound}. See Section~\ref{sec:decomposition} for the achievability proof.
\end{proof}

\begin{remark}[Generalization to $p$-ary Models]
The results in Theorems~\ref{thm:computingcapacity},~\ref{thm:upperbound}, and~\ref{thm:sym_computingcapacity} hold for  $p$-ary models in which $S_{\ell}$'s and channel input/output are in $\mathbb{F}_p$, the desired function is modulo-$p$ addition, and the channel operation is also modulo-$p$ addition. Here $p$ is a prime number. Specifically, in the $p$-ary case, the converse proof of Theorem~\ref{thm:computingcapacity} starts with a slightly different inequality, yet yielding the same result:  $N( (\log_2 p ) R_{\sf comp} - \epsilon_N )$ $\leq H(Y_1^N |X_2^N ) $ $\leq N (\log_2 p ) n_{11}$. One can make the same argument for the other theorems. This will be clearer in the detailed proof that will be presented in Section~\ref{sec:proofofupperbound}.
\end{remark}

\begin{remark}[Comparison to Separation Scheme]
\label{remark:Separation} A baseline strategy for the problem at hand is to let both receivers first fully recover both sources, $S_1$ and $S_2,$ and only then apply the desired computation. Since this strategy separates communication from computation, we refer to it as the {\em separation scheme.} The performance of this strategy is well known: The capacity region for data transmission is simply the intersection of the capacity regions of the two multiple-access channels, one from both transmitters to Receiver 1, the other from both transmitters to Receiver 2. For symmetric models ($n_{11} = n_{22}=n$ and $n_{12}=n_{21}=m$), this evaluates to the message rate region characterized by $R_1  \leq \min (m,n)$, $R_2 \leq \min (m,n)$ and $R_1 + R_2 \leq \max (m,n).$ The corresponding achievable computation rate is simply the maximum symmetric rate point in this region, and thus,
$\frac{R_{\sf comp}^{\sf sep}}{q} \geq \min  \{ \alpha, \frac{ 1}{2}  \}$,
where $\alpha:=\frac{\min(m,n)}{\max(m,n)}$.
This is illustrated in Figure~\ref{fig:computingcapacity}.
For the regime $0 \leq \alpha \leq \frac{1}{2}$, the separation strategy is optimal,
but for $\frac{1}{2} < \alpha \leq 1,$ it is strictly suboptimal. $\square$
\end{remark}

\begin{figure}[t]
\begin{center}
{\epsfig{figure=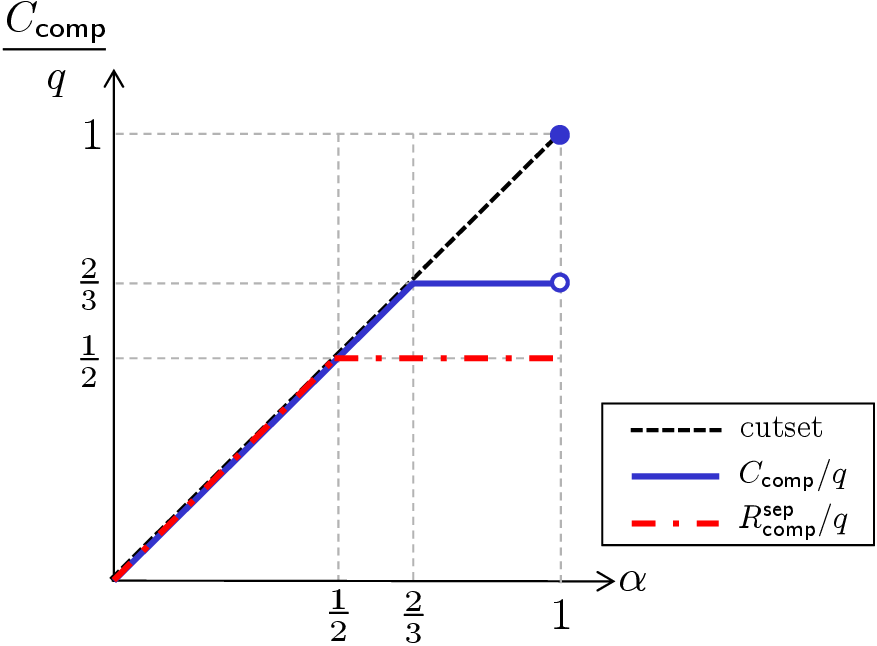, angle=0, width=0.4\textwidth}}
\end{center}
\caption{Computation capacity of the symmetric network ($n := n_{11} = n_{22}$ and $ m:= n_{12} = n_{21}$). Here $\alpha:= \frac{\min(m,n)}{\max(m,n)}$ and $q:=\max(m,n)$.} \label{fig:computingcapacity}
\end{figure}

\subsection{$L$-by-$L$ Network}

We consider a symmetric setting where the only two integer parameters of $(m,n)$ describe the network. Here $n$ indicates the number of signal bit levels from transmitter $\ell$ to receiver $\ell$; and $m$ denotes the number of signal bit levels from transmitter $\ell$ to receiver $\ell' (\neq \ell)$.

\begin{theorem}[$L$-by-$L$ Symmetric Network]
\label{thm:Luser_computingcapacity}
\begin{align*}
R_{\sf comp} \geq
\left\{
  \begin{array}{ll}
    \min \left \{ m, n, \frac{1}{2} \max (n,m) \right \}, & \hbox{$m \neq n$;} \\
    n, & \hbox{$m=n$.}
  \end{array}
\right.
\end{align*}
\begin{align*}
C_{\sf comp} \leq \left\{
  \begin{array}{ll}
    \min \left \{ m, n, \frac{L}{2L-1} \max (n,m) \right \}, & \hbox{$m \neq n$;} \\
    n, & \hbox{$m=n$.}
  \end{array}
\right.
\end{align*}

\end{theorem}
\begin{proof}
See Section~\ref{sec:L-L-symmetric}.
\end{proof}
\begin{remark}
Note that our information-theoretic upper bound approaches the achievable rate as $L$ tends to infinity. As in the $L=2$ case, this result also holds for general $p$-ary models. This will be clarified in Section~\ref{sec:L-L-symmetric}. $\square$
\end{remark}

\section{Proof of Theorem~\ref{thm:upperbound}}
\label{sec:proofofupperbound}

\begin{lemma} \label{lemma:network_class}
For non-degenerate networks ($n_{11} - n_{12} \neq n_{21}  - n_{22}$), there exists $(i,j)$ such that ${\bf G}^{q-n_{ij}} X_i$ can be reconstructed from the pair $(Y_1, Y_2)$.
\end{lemma}
\begin{proof}
See Appendix~\ref{app:lemma2}.
\end{proof}
Using this lemma, we can assume that without loss of generality, ${\bf G}^{q-n_{12}} X_1$ is a function of $(Y_1,Y_2)$.

Our proof includes general $p$-ary models. Starting with Fano's inequality, we get:
\begin{align*}
&N( 3 (\log_2 p) R_{\sf comp} - \epsilon_N) \\
&{\leq} I (S_1^K \oplus S_2^K; Y_1^N) + I (S_1^K \oplus S_2^K; Y_2^N )  + I (S_1^K \oplus S_2^K; Y_2^N)  \\
&\overset{(a)}{\leq} [H(Y_1^N) - H (Y_1^N| S_1^K \oplus S_2^K)]  \\
& \;\; + [H(Y_2^N) - H (Y_2^N| S_1^K \oplus S_2^K, Y_1^N)] + I (S_1^K \oplus S_2^K; Y_2^N) \\
& \leq H(Y_1^N) + H(Y_2^N)  \\
& \;\; - H (Y_1^N, Y_2^N| S_1^K \oplus S_2^K)  + I (S_1^K \oplus S_2^K; Y_2^N, S_2^K)   \\
& \overset{(b)}{=} H(Y_1^N) + H(Y_2^N) \\
& \;\; - H(Y_1^N, Y_2^N| S_1^K \oplus S_2^K)  + H ( \left\{{\bf G}^{q-n_{12}} X_{1i} \right\}_{i=1}^N |S_2^K ) \\
& \overset{(c)}{\leq } H(Y_1^N) + H(Y_2^N) \leq \sum [H(Y_{1i}) + H( Y_{2i})] \\
& \leq N (\log_2 p) [\max (n_{11},n_{21}) + \max (n_{12},n_{22})]
\end{align*}
where $(a)$ follows from the fact that conditioning reduces entropy; $(b)$ follows from the fact that $S_2^K$ is independent of $S_1^K \oplus S_2^K$, and that $X_2^N$ is a function of $S_2^K$; and $(c)$ follows from $H(Y_1^N, Y_2^N| S_1^K \oplus S_2^K)  \geq H ( \left\{ {\bf G}^{q-n_{12}} X_{1i} \right\}_{i=1}^N |S_2^K )$  (see Claim~\ref{claim0} below). Hence, we get the desired bound.

\begin{claim}
\label{claim0}
\begin{align}
H(Y_1^N, Y_2^N| S_1^K \oplus S_2^K)  \geq H \left( \left\{ {\bf G}^{q-n_{12}} X_{1i} \right\}_{i=1}^N |S_2^K \right).
\end{align}
\end{claim}
\begin{proof}
\begin{align*}
& H(Y_1^N, Y_2^N| S_1^K \oplus S_2^K)  - H \left( \left\{ {\bf G}^{q-n_{12}} X_{1i} \right\}_{i=1}^N |S_2^K \right) \\
& \overset{(a)}{=} H \left(Y_1^N, Y_2^N,  \left\{ {\bf G}^{q-n_{12}} X_{1i} \right\}_{i=1}^N | S_1^K \oplus S_2^K \right) \\
& \;\;  - H \left( \left\{{\bf G}^{q-n_{12}} X_{1i} \right\}_{i=1}^N|S_2^K \right)  \\
&\geq H \left( \left\{ {\bf G}^{q-n_{12}} X_{1i} \right\}_{i=1}^N | S_1^K \oplus S_2^K  \right) \\
& \;\; - H \left( \left\{ {\bf G}^{q-n_{12}} X_{1i} \right\}_{i=1}^N | S_2^K \right) \\
& \overset{(b)}{=} H \left( \left\{ {\bf G}^{q-n_{12}} X_{1i} \right\}_{i=1}^N \right) - H \left( \left\{ {\bf G}^{q-n_{12}} X_{1i} \right\}_{i=1}^N | S_2^K \right) \geq 0.
\end{align*}
where $(a)$ follows from our hypothesis that ${\bf G}^{q-n_{12}}X_1$ is a function of $(Y_1, Y_2)$ and $(b)$ follows from the fact that $X_1^N$ is a function of $S_1^K$ that is independent of $S_1^K \oplus S_2^K$.
\end{proof}

\section{Achievability Proof of Theorem~\ref{thm:sym_computingcapacity}}
\label{sec:decomposition}

In this section, we present a network decomposition theorem that permits to decompose a network into elementary subnetworks. Our achievable strategy is then developed {\em separately} for each subnetwork.
In general, one must expect such a decomposition to entail a loss of optimality.
However, for the case $L=2$ (two transmitters, two receivers), our converse proof in Theorem~\ref{thm:upperbound} implies that our network decomposition does not lead to any loss in performance.

\subsection{Achievability via Network Decomposition}

In the sequel, we will use the terminology $(m,n)$ {\it model}  to mean a symmetric $L$-user ADT network.
In the case $m<n,$ each transmitter and each receiver has $n$ levels, and at each receiver, the last $m$ levels
are being interfered in a modulo-sum fashion by the $m$ first levels of each interfering transmitter, as illustrated in Figure~\ref{fig:model}. In the case $m>n,$ each transmitter and each receiver has $m$ levels. For each transmitter, its first $n$ levels connect to the last $n$ levels of its corresponding receiver. Furthermore, at each receiver, the last $n$ levels associated with its corresponding transmitter are being interfered in a modulo-sum fashion by the last $n$ levels of each interfering transmitter; at the remaining top $(m-n)$ levels, collision occurs only across interfering transmitters.
As illustrated by Figure~\ref{fig:model}, there is a natural representation of ADT models in terms of directed graphs, and we will occasionally refer to this representation in the sequel.
\begin{definition}[Network Concatenation]
The $L$-user concatenated model $(m,n) \times (\tilde{m},\tilde{n})$ is constructed from the $L$-user $(m,n)$ model and the $L$-user $(\tilde{m},\tilde{n})$
by merging Transmitter $\ell$ from the $L$-user $(m,n)$ model with Transmitter $\ell$ from the $L$-user $(\tilde{m},\tilde{n})$ model, and by also merging
Receiver $\ell$ from the $L$-user $(m,n)$ model with Receiver $\ell$ from the $L$-user $(\tilde{m},\tilde{n})$ model,  for $\ell=1, 2, \ldots, L.$
As a shorthand, we will also use the notation $(m,n)^k = (m,n) \times (m,n) \times \cdots \times (m,n),$ where there are $k$ terms, i.e., the concatenation of $k$ $L$-user $(m,n)$ models.
\end{definition}
Note that it is straightforward to see that network concatenation is both commutative and associative, e.g., $(m,n) \times (\tilde{m},\tilde{n})$
is exactly the same model as $(\tilde{m},\tilde{n}) \times (m,n).$
\begin{definition}[Network Decomposition]
We say that the $L$-user $(m,n)$ model can be {\em decomposed into the $L$-user concatenated model}  $(m_1,n_1) \times (m_2, n_2),$
denoted as $(m,n) \longrightarrow (m_1,n_1) \times (m_2, n_2),$ if the directed graph corresponding to the $L$-user $(m,n)$ model
is isomorphic to the directed graph corresponding to the $L$-user concatenated model  $(m_1,n_1) \times (m_2, n_2).$
\end{definition}
An example of a network decomposition is given in Figure~\ref{fig:decomposition_example}. The figure graphically proves the fact that, in our just defined notation,
$(2,7) \longrightarrow (0,1)^3 \times (1,2)^2.$ The following theorem establishes a number of more general statements of this type.

\begin{theorem}[Network Decomposition]
\label{thm:networkdecomposition}
The $L$-transmitter $L$-receiver $(m,n)$ network where $m \neq n,$
can be decomposed into the following subnetworks:
\begin{enumerate}
\item[$(1)$]   For any $k\in {\mathbb Z}^+,$
\begin{align*}
(km,kn) \longrightarrow (m,n)^k =  (m,n) \times (m,n) \times \cdots \times (m,n).
\end{align*}
\item[$(2)$]   $(2m+1, 2n+1) \longrightarrow (m,n) \times (m+1, n+1)$.
\item[$(3)$]  For the arbitrary $(m, n)$ model,
\small
\begin{align}
\label{Eq-gap1decomp}
(m, n) \longrightarrow  \left\{
         \begin{array}{ll}
           (r, r+1)^{n-m-a} \times (r+1, r+2)^a, & \hbox{$m < n$;} \\
           (r+1, r)^{m-n-a} \times (r+2, r+1)^a, & \hbox{$m > n$.}
         \end{array}
       \right.
\end{align}
\normalsize
where
\begin{align}
\begin{split}
\label{eq:ra}
r  =   \left\lfloor  \frac{\min\{ m,n \}}{|n-m|} \right\rfloor, \;a  =   \min \{m, n\} \mod |n-m|.
\end{split}
\end{align}
\end{enumerate}
\end{theorem}

\begin{figure}[t]
\begin{center}
{\epsfig{figure=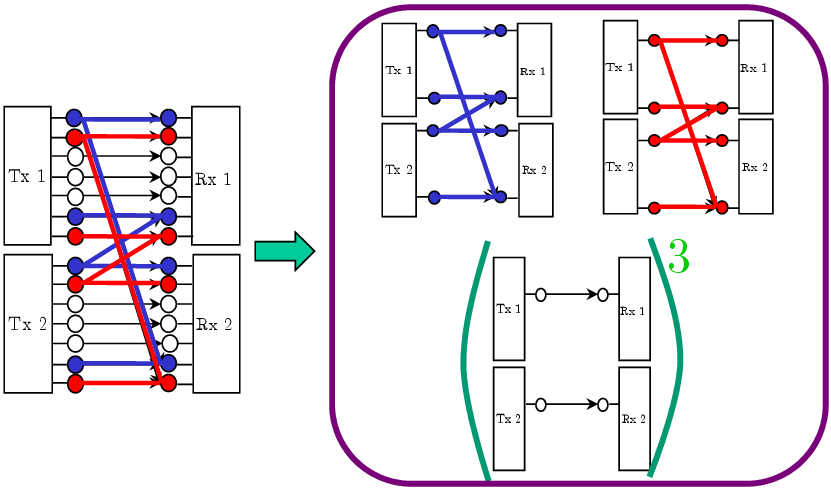, angle=0, width=0.47\textwidth}}
\end{center}
\caption{A network decomposition example of an $(m,n)=(2,7)$ model. From~\eqref{Eq-gap1decomp}, $r=0$ and $a=2$; hence, the decomposition is given by $(2,7) \longrightarrow (0,1)^3 \times (1,2)^2$. } \label{fig:decomposition_example}
\end{figure}

The proof is given in Appendix~\ref{app:proof_NDT}. Here we provide a proof idea with an $(m,n,L)=(2,7,2)$ example, illustrated in Fig.~\ref{fig:decomposition_example}. The idea is to use graph coloring with $|n-m|=5$ colors, identified by integers $\{ 0, 1, 2,3, 4 \}$. At transmitter 1, assign to level $1$ and level $6$ ($=1+ |n-m|$) the color $0$ (blue color in this example). Use exactly the same rule to color the levels of transmitter 2 and receivers 1 and 2. The blue-colored graph represents an independent graph of model $(1,2)$.
Next we assign the color 1 (red color in this example) to level $2$ and level $7$ ($=2 + |n-m|$), for all transmitters and receivers. We then obtain another independent graph of model $(1,2)$ and are left with model $(0,3)$. Obviously the model $(0,3)$ is decomposed into $(0,1)^3$. Therefore, we get $(2,7) \longrightarrow (1,2)^2 \times (0,1)^3$.

In order to prove the achievability in Theorem~\ref{thm:sym_computingcapacity},
we will in particular leverage item (3) in Theorem~\ref{thm:networkdecomposition}.
That is, the decomposition of any (symmetric) network into ``gap-1'' models, i.e.,
models of the form $(r,r+1)$ or $(r+1,r).$
For those models, we can establish the following computation rates:

\begin{lemma}[$L=2$]
\label{lemma:comprates-gap1models}
The following computation rates are achievable:
\begin{enumerate}
\item[$(1)$] For the model $(0,1)$,  $R_{\sf comp} \geq 0$.
\item[$(2)$] For the model $(1,2)$,  $R_{\sf comp} \geq 1$.
\item[$(3a)$] For the model $(r-1,r)$ with $r \geq 3$,  $R_{\sf comp} \geq \frac{2}{3} r$.
\item[$(3b)$] For the model $(r,r-1)$ with $r \geq 3$,  $R_{\sf comp} \geq \frac{2}{3} r$.
\item[$(4)$] For the model $(r, r)$, $R_{\sf comp} \geq r$.
\end{enumerate}
\end{lemma}
\begin{proof}
See Section~\ref{sec:achievability_elementarynetworks}.
\end{proof}

By Theorem~\ref{thm:networkdecomposition} and Lemma~\ref{lemma:comprates-gap1models}, we can now readily prove the achievability of Theorem~\ref{thm:sym_computingcapacity}. By symmetry, we focus on the case of $m \leq n$. The other case of $m \geq n$ is a mirrored image in which the roles of transmitters 1 and 2 are swapped.

For the case of $\alpha=1$, $R_{\sf comp} \geq n$ by Item (4) in Lemma~\ref{lemma:comprates-gap1models}.
For the case of $0 \le \alpha < \frac{1}{2}$, $r=0$ and $a=m$ in~\eqref{eq:ra}; hence, the decomposition is given by
$(m, n) \longrightarrow  (0, 1)^{n-2m} \times (1, 2)^m$.
Thus, using Lemma~\ref{lemma:comprates-gap1models}, the computation rate is
$R_{\sf comp} \geq  0 \cdot (n-2m) + 1 \cdot m= m$.
Next, consider the case of $\frac{1}{2} \le \alpha < \frac{2}{3}$.
Applying the decomposition~\eqref{Eq-gap1decomp},
we find that in this case, $r=1$ and $a=2m-n$:
$(m, n)  \longrightarrow   (1, 2)^{2n-3m} \times (2, 3)^{2m-n}$.
Thus, using Lemma~\ref{lemma:comprates-gap1models}, the computation rate is
$R_{\sf comp} \geq  1 \cdot (2n-3m) + 2 \cdot (2m-n) = m$.
Finally, consider the case of $\alpha \ge \frac{2}{3}$.
Applying the decomposition~\eqref{Eq-gap1decomp},
we find that in this case, $r\ge 2$. So we get
\begin{align*}
R_{\sf comp} & \geq \frac{2}{3} (r+1) (n-m-a) + \frac{2}{3} (r+2) a \\
&= \frac{2}{3} \left \{   r (n-m) + a   + (n-m)\right \} \\
&\overset{(a)}{=} \frac{2}{3} \left \{  m  + (n-m)\right \} = \frac{2}{3} n.
\end{align*}
where $(a)$ is due to~\eqref{eq:ra}. This completes the proof.

\subsection{Proof of Lemma~\ref{lemma:comprates-gap1models}}
\label{sec:achievability_elementarynetworks}

We note that Items (1), (2) and (4) are obvious, and Item (3b) follows from Item (3a),
since without loss of generality, for the multicast problem with $L=2$ users considered here,
the case $(m,n)$ and the case $(n,m)$ are mirror images of each other in which the roles of
transmitters 1 and 2 are swapped.
We here provide an explicit proof of Item (3a), split into three cases.
For notation, the symbols of Transmitter 1 will be denoted by $a_1, a_2, a_3, \ldots$ and the symbols of Transmitter 2 by $b_1, b_2, b_3 \ldots.$
The goal of both receivers is then to recover the modulo sums $a_1 \oplus b_1, a_2 \oplus b_2, \ldots.$
Moreover, we will find it convenient to collect the channel inputs used by transmitters into length-$n$ vectors denoted by ${\bf x}_1$ and ${\bf x}_2,$ respectively.

{\em (i)} The case $r=3 \ell,$ i.e., the $(3 \ell-1, 3 \ell)$ model:
Let us start with the simplest case of $(2,3)$ model.
At Transmitter 1, we send ${\bf x}_{1}Ê= ( a_1, a_2, 0)$
and at Transmitter 2, ${\bf x}_{2}Ê= ( b_2, b_1, 0).$
Clearly, both receivers learn both modulo sums. This code can be extended to all models for the form $(3\ell-1,3\ell)$
(for all positive integers $\ell$), as follows.
We set
$X_{1,3k-2} = a_{2k-1}, X_{1,3k-1} = a_{2k}, X_{1,3k}=0$
and
$X_{2,3k-2} = b_{2k}, X_{2,3k-1} = b_{2k-1}, X_{2,3k}=0,$
for $k= 1, 2, \ldots, \ell.$
Each receiver can reconstruct all $2\ell$ sums $a_k \oplus  b_k$
and thus, the computation rate is $2\ell = \frac{2}{3}(3\ell),$ as claimed.\footnote{
In the solution for the $(2,3)$ model, the symbols $(a_1,b_1)$ at receiver 2 share one-dimensional linear subspace spanned by $(0,1,0)$. In other words, the linear subspace with respect to $a_1$ is \emph{aligned} with the subspace w.r.t $b_1$. In this sense, it is an instance of the important concept of \emph{interference alignment}~\cite{Mohammad:IT08,CadambeJafar:IT08} which has shown the great potential for a variety of applications~\cite{CadambeJafar:IT08,Suh:Allerton,Suh:TRCOM,Wu:ISIT09, KumarRamchandran_MSR:IT12, Suh:ER-MDS,Das:ISIT10,RamakrishnanJafar:Allerton}. But the distinction w.r.t our problem comes from the purpose of alignment. In our problem, the aim of alignment is to form a desired function while minimizing the signal subspace occupied by the source symbols.}

\begin{figure*}
\begin{center}
{\epsfig{figure=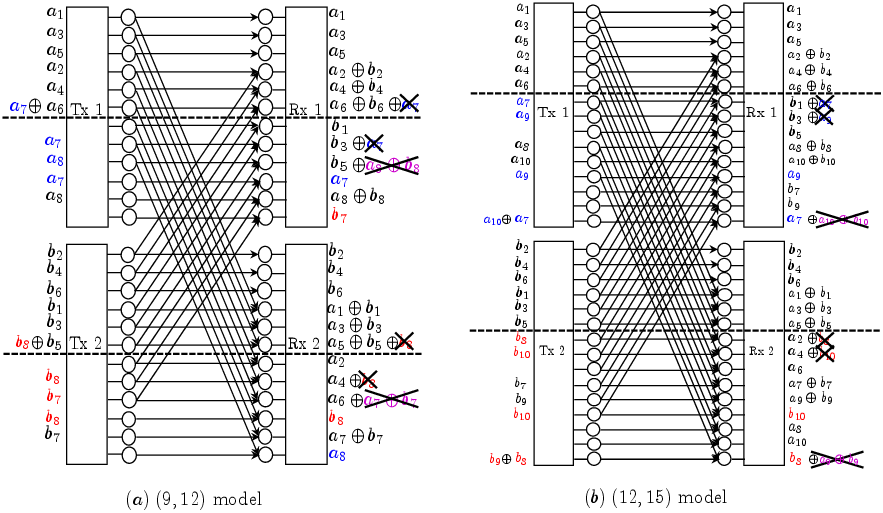, angle=0, width=0.8\textwidth}}
\end{center}
\caption{Explicit code for the $(9,12)$ model and the $(12,15)$ model.} \label{fig:explicit_code}
\end{figure*}

{\em (ii)} The case $r=3 \ell+1,$ i.e., the $(3 \ell, 3 \ell + 1)$ model:
Consider the $(3,4)$ model. Here we consider a code for the $(9,12)$ model (which can be implemented over three channel uses of the $(3,4)$ model). The code is given in Fig.~\ref{fig:explicit_code}$(a)$.
This code enables both receivers to recover all $8$ computations,
which means a computation rate of $8/3 = \frac{2}{3}\cdot4$ per channel use, as claimed in the lemma.

Now, consider the general $(3\ell,3\ell+1)$ model, for any positive integer $\ell\ge2.$ In this case, we find that the $(6,9)$ model acts as a building block, and hence we will first develop a code for the $(6,9)$ code.

{\em The $(6,9)$ code:}  This network falls into the $(2,3)$ model category as it can be decomposed into three orthogonal $(2,3)$ networks.
Since it will have a prominent role in the remainder of the proof, we explicitly
write out the code. Here, each transmitter has 6 symbols.
At Transmitter 1, we send ${\bf x}_{1}Ê= ( a_1, a_3, a_5, a_2, a_4, a_6, 0, 0, 0)$
and at Transmitter 2, ${\bf x}_{2}Ê= ( b_2, b_4, b_6, b_1, b_3, b_5, 0, 0, 0).$
It is easy to verify that all 6 modulo sums can be recovered.
The key to observe is that the last 3 inputs at both transmitters are all zero.


Now, let us go back to the general case: $(3\ell,3\ell+1)$ model, $\ell\ge2.$ Again, we code over $3$ channel uses.
By network decomposition, any code for the $(9\ell, 9\ell+3)$ model
can be implemented over $3$ channel uses of the $(3\ell,3\ell+1)$ model.
But for the $(9\ell, 9\ell+3)$ model,
we {\em split}  the network into multiple parts:
On the first $9$ vertices of each transmitter and receiver, we implement the $(6,9)$ code from above,
giving 6 computations.
But since in that code, neither transmitter uses the last 3 inputs, this leaves the remaining
$9\ell+3-9 = 9(\ell-1)+3$ vertices of the network completely unaffected at all transmitters and receivers.
Hence, we can repeat this step: on the next 9 vertices, again implement the $(6,9)$ code from above,
giving another 6 computations.
We do this step exactly $\ell-1$ times, leading to $6(\ell -1)$ computations.
At that point, we are left with $9\ell+3 - 9(\ell-1) = 12$ vertices
at each transmitter and each receiver.
On these, we implement the code from Fig.~\ref{fig:explicit_code}$(a)$, giving us another 8 computations. This gives a total of $6(\ell-1) + 8 = 2 ( 3\ell +1 )$ computations. Thus, per channel use, the computation rate is $\frac{2}{3} ( 3\ell +1 ),$ as claimed.


{\em (iii)} The case $r=3 \ell+2,$ i.e., the $(3 \ell + 1 , 3 \ell + 2)$ model:
Consider the $(4,5)$ model. Here we consider a code for the $(12,15)$ model (which can be implemented over three channel uses
of the $(4,5)$ model). The code is given in Fig.~\ref{fig:explicit_code}$(b)$.
This code enables both receivers to recover all $10$ computations,
which means a computation rate of $10/3 = \frac{2}{3}\cdot5$ per channel use, as claimed in the lemma.

Now, consider the general $(3\ell+1,3\ell+2)$ model, for any positive integer $\ell\ge2.$
Again, we code over $3$ channel uses.
By network decomposition, any code for the $(9\ell+3, 9\ell+6)$ model
can be implemented over $3$ channel uses of the $(3\ell+1,3\ell+2)$ model.
But for the $(9\ell+3, 9\ell+6)$ model,
we {\em split}  the network into multiple parts:
On the first $9$ vertices of each transmitter and receiver, we implement the $(6,9)$ code from above,
giving 6 computations.
But since in that code, neither transmitter uses the last 3 inputs, this leaves the remaining
$9\ell+6-9 = 9(\ell-1)+6$ vertices of the network completely unaffected at all transmitters and receivers.
Hence, we can repeat this step: on the next 9 vertices, again implement the $(6,9)$ code from above,
giving another 6 computations.
We do this step exactly $\ell-1$ times, leading to $6(\ell -1)$ computations.
At that point, we are left with $9\ell+6 - 9(\ell-1) = 15$ vertices
at each transmitter and each receiver.
On these, we implement the code from Fig.~\ref{fig:explicit_code}$(b)$,
giving us another 10 computations. This gives a total of $6(\ell-1) + 10 = 2 ( 3\ell +2 )$ computations.
Thus, per channel use, the computation rate is $\frac{2}{3} ( 3\ell +2 ).$

\section{Proof of Theorem~\ref{thm:Luser_computingcapacity}}
\label{sec:L-L-symmetric}

\subsection{Achievability Proof}
\label{sec:Luser_achievability}

The idea is to combine the network decomposition in Theorem~\ref{thm:networkdecomposition} and achievability proof for elementary subnetworks.

\begin{figure}[t]
\begin{center}
{\epsfig{figure=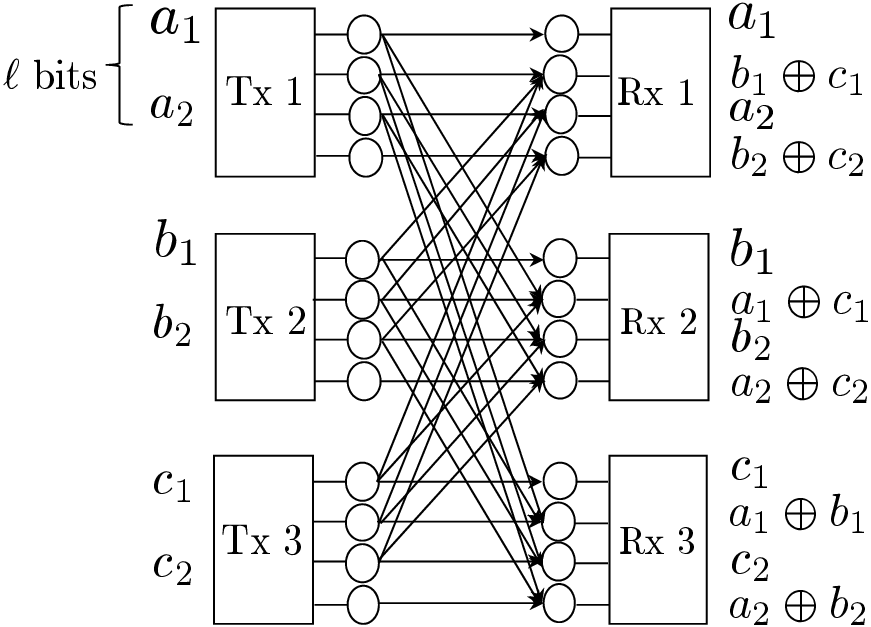, angle=0, width=0.35\textwidth}}
\end{center}
\caption{Achievable scheme for the $(2\ell-1,2\ell)$ model.
} \label{fig:Luser_even}
\end{figure}

\begin{lemma}[$L \geq 3$]
\label{lemma:Lusercomprates-gap1models}
The following computation rates are achievable:
\begin{enumerate}
\item[$(1)$] For the model $(0,1)$ or $(1,0)$,  $R_{\sf comp} \geq 0$.
\item[$(2a)$] For the model $(r-1,r)$ with $r \geq 2$,  $R_{\sf comp} \geq \frac{1}{2} r$.
\item[$(2b)$] For the model $(r,r-1)$ with $r \geq 2$,  $R_{\sf comp} \geq \frac{1}{2} r$.
\item[$(3)$] For the model $(r, r)$, $R_{\sf comp} \geq r$.
\end{enumerate}
\end{lemma}
\begin{proof}
The items $(1)$ and $(3)$ are straightforward.

{\em (i)} $(r-1,r)$ model: We consider two cases: $r=2 \ell$ and $r=2\ell+1$. Fig.~\ref{fig:Luser_even} shows an achievable scheme when $r=2\ell=2 \cdot 2$ and $L=3$. Each transmitter uses odd-numbered vertices to send $\ell$ symbols. The special structure of symmetric networks enables each receiver to get clean symbols on odd-numbered vertices while receiving partially-aligned functions on even-numbered ones. For example, receiver 1 gets $(a_1,a_2)$ on the first and third vertices; $(b_1 \oplus c_1,b_2 \oplus c_2)$ on the second and fourth vertices. Note that two resource levels are consumed to compute one desired function. Therefore, this gives a computation rate of $\frac{1}{2} r$. Obviously this can be applied to an arbitrary value of $L$.

\begin{figure}[t]
\begin{center}
{\epsfig{figure=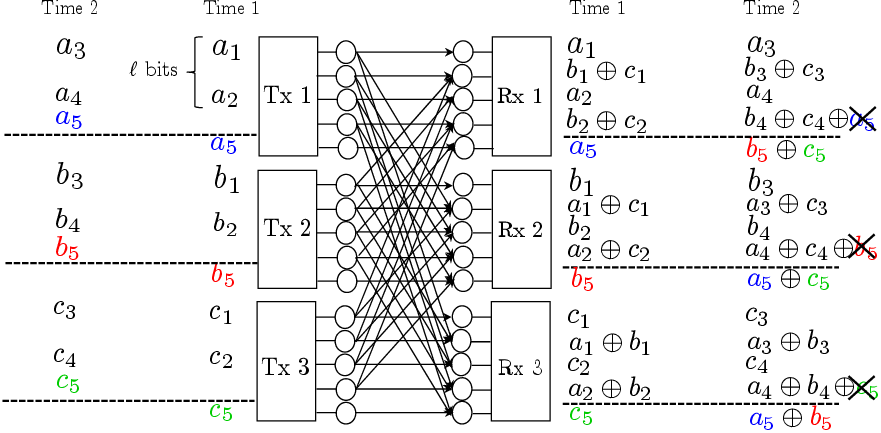, angle=0, width=0.47\textwidth}}
\end{center}
\caption{Achievable scheme for the $(2\ell, 2 \ell+1)$ model.
} \label{fig:Luser_odd}
\end{figure}

An explicit scheme for the case of $r=2 \ell+1=2 \cdot 2 +1$ and $L=3$ is given in Fig.~\ref{fig:Luser_odd}. If we followed the same approach as in the case of $r=2 \ell$, each receiver would be left with one empty vertex in the last level. In this example, receiver 1 would get $(a_1, b_1 \oplus c_1, a_2, b_2 \oplus c_2)$ on the 1st, 2nd, 3rd, and 4th vertices, while the last bottom vertex is empty. In order to make an efficient resource utilization, we invoke an idea of vector coding. First we implement the same code in time 2. Now each transmitter will exploit the two empty vertices to provide one more computation. In time 1, on the last vertex, each transmitter sends a new fresh symbol. In this example, transmitters 1, 2, 3 send $a_5$, $b_5$, $c_5$ respectively. However, this transmission does not give the computation of $a_5 \oplus b_5 \oplus c_5$ yet. In time 2, on the second last vertex, each transmitter re-sends the symbol that was sent on the last vertex in time 1. Receiver 1 can then get $b_5 \oplus c_5$, thus achieving $a_5 \oplus b_5 \oplus c_5$. While the transmission of $a_5$ causes interference to $b_4 \oplus c_4$ on the second last vertex, we can resolve this conflict by using the $a_5$ (that was already received in time 1) as side information. Similarly the desired function can be computed at the other receivers. This way, we can achieve $\frac{2\ell+ 1 }{2} = \frac{r}{2}$. The same strategy can be applied to an arbitrary value of $L$.

{\em (ii)} $(r,r-1)$ model: Consider two cases: $r=2 \ell$ and $r=2\ell+1$. For both cases, the coding strategies are the same as those in the $(r-1,r)$ model. In the first case, each transmitter sends $\ell$ symbols only on odd-numbered vertices. One can then readily see that all the receivers can compute $\ell$ sums. In the second case, we use two time slots. In each time, on the upper $2 \ell $ vertices of transmitters and receivers, we implement the $(2 \ell, 2\ell -1)$ code from above. On top of this, each transmitter sends one additional fresh symbol to provide one more computation. Specifically each transmitter sends the $(\ell +1)$th symbol on the last vertex in time 1 and re-sends the symbol on the second last vertex in time 2. One can check that each receiver can compute $2\ell +1$ sums, thus yielding a computation rate of $\frac{r}{2}$.
\end{proof}

Using Theorem~\ref{thm:networkdecomposition} and Lemma~\ref{lemma:Lusercomprates-gap1models}, we can now prove the achievability. For the case of $\alpha=1$, $R_{\sf comp} \geq n$ by Item (3) in Lemma~\ref{lemma:Lusercomprates-gap1models}.
For $0 \le \alpha < \frac{1}{2}$,~\eqref{Eq-gap1decomp} gives $r=0$ and $a=m$, thus the decomposition is given by
$(m, n) \longrightarrow  (0, 1)^{n-2m} \times (1, 2)^{m}$.
Therefore, using Lemma~\ref{lemma:Lusercomprates-gap1models}, $R_{\sf comp} \geq  0 \cdot (n-2 m) + 1 \cdot m= m$. Next, consider the case of $\frac{1}{2} \le \alpha < \frac{2}{3}$. Using~\eqref{Eq-gap1decomp}, we find that $r=1$ and $a=2m-n$, hence, the decomposition is given by $(m, n)  \longrightarrow   (1, 2)^{2n-3m} \times (2, 3)^{2m-n}$. Using Lemma~\ref{lemma:Lusercomprates-gap1models}, $R_{\sf comp} \geq  1 \cdot (2n-3m) + \frac{3}{2} \cdot (2m-n) = \frac{1}{2} n$.
Finally, consider the case of $\alpha \ge \frac{2}{3}$.
From~\eqref{Eq-gap1decomp}, we know that $r\ge 2$. So we get
$R_{\sf comp} \geq \frac{1}{2} n$.

The other case of $m>n$ similarly follows. For  $0 \le \alpha < \frac{1}{2}$,~\eqref{Eq-gap1decomp} gives $r=0$ and $a=n$, thus the decomposition is given by
$(m, n) \longrightarrow  (1, 0)^{m-2n} \times (2, 1)^n$. With Lemma~\ref{lemma:Lusercomprates-gap1models}, we get $R_{\sf comp} \geq  0 \cdot (m-2n) + 1 \cdot n= n$. For $\frac{1}{2} \le \alpha < \frac{2}{3}$,~\eqref{Eq-gap1decomp} gives $r=1$ and $a=2n-m$. So the decomposition is given by $(m, n)  \longrightarrow   (2, 1)^{2m-3n} \times (3, 2)^{2n-m}$. Using Lemma~\ref{lemma:Lusercomprates-gap1models}, we can achieve
$R_{\sf comp} \geq  1 \cdot (2m-3n) + \frac{3}{2} \cdot (2n-m) = \frac{1}{2} m$.
Lastly, for $\alpha \ge \frac{2}{3}$,~\eqref{Eq-gap1decomp} gives $r\ge 2$. So we get
$R_{\sf comp} \geq \frac{1}{2} m$.

\subsection{Proof of Upper Bound}
\label{sec:Luser_upperbound}

The upper bound is a generalized version of the $2\times 2$ case bound~\eqref{eq:upper2}. Hence, the bounding technique is almost the same, but it involves complicated notations. Let ${\cal Y}:= \{ Y_{\ell}^N \}_{\ell=1}^L$ be a collection of all the received signals. Let $S:= \bigoplus_{\ell=1}^L S_{\ell}^K$ be the $K$-dimensional desired sum function. Let $\bar{\cal S}_{\ell}:= \{ S_j^K \}_{j=1}^{L} \setminus S_{\ell}^K$ be a collection of all the sources that excludes the $\ell$th user's source. Using these notations and starting with Fano's inequality, we get:

\begin{align*}
&N(  (\log_2 p) (2L-1) R_{\sf comp} - \epsilon_N) \\
&\leq \sum_{ \ell = 1}^L I ( S ; Y_\ell^N) + (L-1) I ( S ; Y_1^N)\\
& \overset{(a)}{\leq} \sum_{ \ell = 1}^L \left[ H ( Y_\ell^N) - H ( Y_\ell^N| S,  \{ Y_{j}^{N} \}_{j=1}^{\ell -1}  ) \right ] + (L-1) I ( S ; Y_1^N) \\
& \overset{(b)}{\leq} \sum_{ \ell = 1}^L  H ( Y_\ell^N)   -   H ( {\cal Y} | S ) + \sum_{ \ell = 2}^{L} I ( S ; Y_1^N,   \bar{\cal S}_{\ell} ) \\
&   \overset{(c)}{=} \sum_{ \ell = 1}^L  H (Y_\ell^N) -   H ( {\cal Y}| S) +  \sum_{ \ell = 2}^{L} H \left (  Y_1^N | \bar{\cal S}_{\ell} \right)  \\
&  \overset{(d)}{=} \sum_{ \ell = 1}^L  H (Y_\ell^N) -   H ( {\cal Y}| S) +  \sum_{ \ell = 2}^{L} H \left (  \{ {\bf G}^{ q- n_{\ell 1}} X_{\ell i} \}_{i=1}^{N} | \bar{\cal S}_{\ell} \right) \\
& \overset{(e)}{\leq} \sum_{ \ell = 1}^L  H ( Y_\ell^N) \leq (\log_2 p) N L \max (m,n)
\end{align*}
where $(a)$ follows from the fact that conditioning reduces entropy; $(b)$ follows from a chain rule and the non-negativity of mutual information;  $(c)$ follows from the fact that ${\bar{\cal S}}_{\ell}$ is independent of $S$ due to the mutual independence of $S_{\ell}^{K}$'s;
 $(d)$ follows from the fact that  $X_{\ell' i}$ is a function of $\bar{\cal{S}}_{\ell}$ $\forall \ell' \neq \ell$;
and $(e)$ follows from $H({\cal Y} |S) \geq  \sum_{ \ell = 2}^{L} H \left (  \{ {\bf G}^{ q- n_{\ell 1}} X_{\ell i} \}_{i=1}^{N} | \bar{\cal S}_{\ell} \right)$ (see Claim~\ref{claim1} below). Hence, $R_{\sf comp} \leq \frac{L}{2L-1} \max (m,n)$.

\begin{claim}
\label{claim1}
\begin{align}
H({\cal Y} |S) \geq  \sum_{ \ell = 2}^{L} H \left (  \{ {\bf G}^{ q- n_{\ell 1}} X_{\ell i} \}_{i=1}^{N} | \bar{\cal S}_{\ell} \right).
\end{align}
\end{claim}
\begin{proof}
\begin{align*}
& H ( {\cal Y}| S) - \sum_{ \ell = 2}^{L} H \left (  \{ {\bf G}^{ q- n_{\ell 1}} X_{\ell i} \}_{i=1}^{N} | \bar{\cal S}_{\ell} \right) \\
& \overset{(a)}{=} H \left( {\cal Y},  \{ {\bf G}^{ q- n_{2 1}} X_{2 i} \}_{i=1}^{N} , \cdots, \{ {\bf G}^{ q- n_{L 1}} X_{L i} \}_{i=1}^{N} | S \right) \\
& \;\; - \sum_{ \ell = 2}^{L} H \left (  \{ {\bf G}^{ q- n_{\ell 1}} X_{\ell i} \}_{i=1}^{N}  | \bar{\cal S}_{\ell} \right)  \\
& \overset{(b)}{\geq}H \left(  \{ {\bf G}^{ q- n_{2 1}} X_{2 i} \}_{i=1}^{N} , \cdots, \{ {\bf G}^{ q- n_{L 1}} X_{L i} \}_{i=1}^{N} | S \right)  \\
& \;\; - \sum_{ \ell = 2}^{L} H \left (  \{ {\bf G}^{ q- n_{\ell 1}} X_{\ell i} \}_{i=1}^{N}  | \bar{\cal S}_{\ell} \right) \\
&\overset{(c)}{=}   \sum_{\ell =2 }^{L} H \left (  \{ {\bf G}^{ q- n_{\ell 1}} X_{\ell i} \}_{i=1}^{N}  \right) - \sum_{ \ell = 2}^{L} H \left (  \{ {\bf G}^{ q- n_{\ell 1}} X_{\ell i} \}_{i=1}^{N}  | \bar{\cal S}_{\ell} \right) \\
& \geq 0
\end{align*}
where $(a)$ follows from the fact that $X_{\ell i}$ is a function of $\{ Y_{ji} \}_{j=1}^{L}$, $\forall \ell, i$ (see Claim~\ref{claim:Luserfunctionalrelation} below); $(b)$ follows from the non-negativity of the entropy and the fact that $(S_2^K,\cdots,S_L^K)$ is independent of $S$; $(c)$ follows from a chain rule, and the fact that  $X_{\ell i}$ is independent of $S$ and  $(X_{(\ell+1)i},\cdots, X_{L i})$'s due to the mutual independence of $S_{\ell}^{K}$'s.
\end{proof}

\begin{claim}
\label{claim:Luserfunctionalrelation}
For $m \neq n$, $X_{\ell}$ is a function of $ \{ Y_{j} \}_{j=1}^{L}$, $\ell =1,\cdots, L$.
\end{claim}
\begin{proof}
By symmetry, it suffices to consider the case of $\ell=1$. Consider the case of $m < n$. From~\eqref{eq:receivedsignals}, we get
\begin{align*}
\sum_{j=1}^{L}Y_j &= \left \{ {\bf I} + (L - 1 ) {\bf G}^{n-m} \right\} \sum_{j=1}^{L}X_j;\\
\sum_{j=2}^{L}Y_j &= \left \{ {\bf I} + (L - 1 ) {\bf G}^{n-m} \right\} \sum_{j=1}^{L}X_j \\
& \;\;\;\;  - \left( X_1 + {\bf G}^{n-m} \sum_{j=2}^{L} X_j \right) \\
&= (L-1) {\bf G}^{n-m} \sum_{j=1}^{L}X_j +   \left \{ {\bf I} - {\bf G}^{n-m} \right \}  \sum_{j=2}^{L}X_j.
\end{align*}
The above two equations can be written as follows respectively:
\begin{align*}
\sum_{j=1}^{L}X_j &= {\bf A}^{-1} \sum_{j=1}^{L}Y_j \\
\sum_{j=2}^{L}X_j &= {\bf B}^{-1} \left[ \sum_{j=2}^{L}Y_j -  (L-1) {\bf G}^{n-m} \left \{ {\bf A}^{-1} \sum_{j=1}^{L}Y_j \right \} \right]
\end{align*}
where ${\bf A}:= {\bf I} + ( L -1 ) {\bf G}^{n-m}$ and ${\bf B}:= {\bf I} - {\bf G}^{n-m}$. Note that when $m \neq n$, both of ${\bf A}$ and ${\bf B}$ are invertible. So $X_1$ can be computed from the above two equations, and thus $X_1$ is a function of $ \{Y_{j} \}_{j=1}^{L}$. Similarly we can show this for the case of $m>n$. Notice that this proof holds for the general $p$-ary model.
\end{proof}

\section{Discussion}
\label{sec:discussion}

\subsection{On the Optimality of Network Decomposition}

A recent work in~\cite{SuhGastpar:SPAWC2013} has shown the optimality of our network-decomposition-based approach for two other problem settings under the 2-user ADT symmetric network: (1) the two-unicast problem in which each receiver wishes to decode a message from its corresponding transmitter; (2) the classical multicast problem in which each receiver wants to decode every message. It would be interesting to explore the optimality of the decomposition approach for more general problem settings.

\subsection{Asymmetric Networks}
For the general asymmetric network, the computation capacity is not characterized even for the $L=2$ case. This is mainly because we found it difficult to develop a network decomposition theorem for the general asymmetric network. Moreover, we conjecture that our upper bound may be loose for some asymmetric network. This conjecture comes from our inspection on the example of $(n_{12},n_{11}; n_{21}, n_{22}) = (3,3;4,5)$. In this example, our upper bound gives $C_{\sf comp} \leq \frac{\max(n_{11}+n_{21}) + \max (n_{22},n_{12})}{3} = 3$. On the other hand, when applying all of the achievability techniques developed in Theorem~\ref{thm:sym_computingcapacity}, we could achieve only a computation rate of $\frac{8}{3}$. It is expected that both of new inner and upper bounds are needed for the asymmetric network.

\subsection{A Class of Linear Deterministic Networks}
As mentioned earlier, although the ADT network setting that we considered is somewhat specialized, the results in Theorems~\ref{thm:computingcapacity},~\ref{thm:upperbound},~\ref{thm:sym_computingcapacity}, and~\ref{thm:Luser_computingcapacity} apply to general $p$-ary models. Furthermore, we expect that the results can be useful for a more general class of linear deterministic networks in which even the channel transfer matrices are of an arbitrary form. In particular, the upper-bound technique can be readily applied to the general setting, although it may not guarantee the optimality. As for achievability, it would be an interesting future work to discover elementary subnetworks (if any) that can constitute an original network without loss of optimality.

\subsection{Multi-hop Networks}

In~\cite{Ramamoorthy:ISIT,RaiDey:it12,RamamoorthyLangberg:journal}, function multicasting has been explored in the context of multi-hop networks. In particular, Rai and Dey in~\cite{RaiDey:it12} found some interesting equivalence relationship between sum-networks and multiple-unicast networks. Due the relationship, it has been believed that the sum-network problem is as hard as the multiple-unicast problem, and indeed the computation capacity of the sum-network has been open. For two-source $L$-destination or $L$-source two-destination networks, the computation capacity was established only when the entropy of each source is constrained to be 1~\cite{RamamoorthyLangberg:journal}.

One natural next step is exploiting the insights developed in this work, to characterize necessary and sufficient conditions of two-source two-destination multi-hop networks when the entropy of each source is limited by 2.

\subsection{Role of Feedback for Computation}

The role of feedback for computation has initially been studied in~\cite{SuhGoelaGastpar:isit12} where it is shown that feedback can increase the computation rate. Interestingly the feedback gain is shown to be significant - qualitatively similar to the gain in the two-user Gaussian interference channel~\cite{SuhTse}.
However, the result of~\cite{SuhGoelaGastpar:isit12} relies on a separation approach that naturally comes in the course of characterizing the feedback multicast capacity. Recently~\cite{SuhGastpar:isit13} has established the exact feedback computation capacity of the $2$-by-$2$ ADT network considered herein to show that the feedback gain can be more significant. It would be interesting to explore this feedback gain under more realistic scenarios where feedback is offered through rate-limited bit-piped links~\cite{AlirezaSuhAves} or through the corresponding backward communication network~\cite{SuhTse:isit12}.

\section{Conclusion}
\label{sec:conclusion}

We have established the computation capacity of a two-transmitter two-receiver ADT symmetric network where each receiver wishes to compute a modulo-2-sum function of two Bernoulli sources generated at the two transmitters. For the $L$-user case, we also characterized the computation capacity in the limit of $L$. In the process of obtaining these results, we derived new upper bounds and established a network decomposition theorem that provides a conceptually-simple achievability proof. We expect that the network-decomposition-based framework would give insights into solving more general network problems.

\appendices

\section{Proof of Lemma~\ref{lemma:network_class}}
\label{app:lemma2}

Consider the following cases: (1) $(n_{12} \leq n_{11},  n_{21} \leq n_{22})$; (2) $(n_{12} \geq n_{11}, n_{21} \geq n_{22})$; (3) $(n_{12} \leq n_{11}, n_{21} \geq n_{22})$; and (4) $(n_{12} \geq n_{11}, n_{21} \leq n_{22})$.
Note that Case (1) and Case (2) are symmetric, so are Case (3) and Case (4). Hence, only for Case (1) and (3), we will show that there exists $(i,j)$ such that ${\bf G}^{q-n_{ij}}X_i$ can be reconstructed from the pair $(Y_1,Y_2)$.

\emph{Case (1)} $(n_{12} \leq n_{11}, n_{21} \leq n_{22})$:
Since $n_{11} \geq n_{12}$, there is no reason to send information on the last $q-n_{11}$ levels in $X_1$. Note that any signal on the levels would be wiped out in ${\bf G}^{q-n_{11}} X_1$. Hence we focus on the case in which $X_1=[ \tilde{X}_1; {\bf 0}_{q-n_{11}}]$. Now let us consider
\begin{align*}
Y_1 \oplus {\bf G}^{n_{22}-n_{21}} Y_2 &=  ( {\bf G}^{q-n_{11}}  \oplus {\bf G}^{q + n_{22} -n_{21} - n_{12}} ) X_1 \\
&= \left[
         \begin{array}{c}
           {\bf 0}_{q-n_{11}} \\
           \left \{ {\bf I}_{n_{11}}  \oplus {\bf G}_{n_{11}}^{n_{11} - n_{12} - ( n_{21} -n_{22}) } \right \} \tilde{X}_1 \\
         \end{array}
       \right]
\end{align*}
where the second equality follows from $X_1=[ \tilde{X}_1; {\bf 0}_{q-n_{11}}]$. Here ${\bf G}_{n_{11}}$ denotes the $n_{11}$-by-$n_{11}$ shift matrix. Since $n_{11} - n_{12} \neq n_{21} - n_{22}$, ${\bf G}_{n_{11}}^{n_{11} - n_{12} - ( n_{21} -n_{22}) } \neq {\bf I}_{n_{11}}$ and therefore ${\bf I}_{n_{11}}  \oplus {\bf G}_{n_{11}}^{n_{11} - n_{12} - ( n_{21} -n_{22}) }$ is invertible, implying that $\tilde{X}_1$ can be reconstructed from $(Y_1,Y_2)$. Hence, ${\bf G}^{q-n_{11}}X_1$ can also be reconstructed from $(Y_1,Y_2)$.

\emph{Case (3)} $(n_{12} \leq n_{11}, n_{21} \geq n_{22})$: First consider the case of $n_{21} - n_{22} > n_{11} - n_{12}$. In this case, we get:
\begin{align*}
{\bf G}^{n_{21}-n_{22}} Y_1 \oplus  Y_2 &= ( {\bf G}^{q- n_{11} + n_{21}-n_{22}} \oplus {\bf G}^{q-n_{12}} ) X_1 \\
 &= \left[
         \begin{array}{c}
           {\bf 0}_{q-n_{11}} \\
           ( {\bf G}_{n_{11}}^{ n_{21} -n_{22} } \oplus {\bf G}_{n_{11}}^{ n_{11} -n_{12} } ) \tilde{X}_1 \\
         \end{array}
       \right].
\end{align*}
Since $n_{11} - n_{12} \neq n_{21} - n_{22}$, ${\bf G}_{n_{11}}^{n_{21} - n_{22}} \neq {\bf G}_{n_{11}}^{n_{11} - n_{12}}$.
So ${\bf G}_{n_{11}}^{ n_{11} -n_{12} } \tilde{X}_1$ can be reconstructed from $(Y_1,Y_2)$. Therefore, ${\bf G}^{ q -n_{12} } {X}_1$ can also be reconstructed from $(Y_1,Y_2)$.
For the other case $n_{21} - n_{22} < n_{11} - n_{12}$, we get:
\begin{align*}
{\bf G}^{n_{11}-n_{12}} Y_1 \oplus  Y_2 &= ( {\bf G}^{q- n_{21} + n_{11}-n_{12}} \oplus {\bf G}^{q-n_{22}} ) X_2 \\
&= \left[
         \begin{array}{c}
           {\bf 0}_{q-n_{21}} \\
           ( {\bf G}_{n_{21}}^{ n_{11} -n_{12} } \oplus {\bf G}_{n_{21}}^{ n_{21} -n_{22} } ) \tilde{X}_2 \\
         \end{array}
       \right].
\end{align*}
Since $n_{11} - n_{12} \neq n_{21} - n_{22}$, ${\bf G}_{n_{21}}^{n_{11} - n_{12}} \neq {\bf G}_{n_{21}}^{n_{21} - n_{22}}$. So  ${\bf G}_{n_{21}}^{ n_{21} -n_{22} } \tilde{X}_2$ can be reconstructed from $(Y_1,Y_2)$. Therefore ${\bf G}^{ q -n_{22} } {X}_2$ can also be reconstructed from $(Y_1,Y_2)$.

\section{Proof of Theorem~\ref{thm:networkdecomposition}}
\label{app:proof_NDT}

The $L$-user $(m,n)$ network studied in this paper is naturally represented by a bi-partite graph as in Figure~\ref{fig:model}.
Each transmitter and each receiver is represented by $\max (n,m)$ vertices.
For the argument developed here, we consider the undirected version of this graph.
First, we attach labels to each vertex.
In particular, at each transmitter and each receiver separately, we label the vertices by the integers
$\{ 0, 1, \ldots, \max (n,m)-1\}$ from top to bottom.
The following observation is key:

\begin{lemma}\label{Lemma-graphcoloring-nopath}
In an $L$-user $(m,n)$ model, for any two vertices, if their respective labels $u$ and $v$
satisfy $(u -v) \mod |n-m| \not= 0,$ then there is no path between them.
\end{lemma}

\begin{proof}
Consider any edge in the graph representation of our network (as illustrated e.g. in  Figure~\ref{fig:model}).
Suppose that the two vertices connected by this edge have labels $a$ and $b.$
Then, the structure of the graph implies that we must have $(a -b) \mod |n-m| = 0.$
Now, suppose that there is a path between two arbitrary vertices with labels $u_1$ and $u_\ell,$
passing through vertices with labels $u_2, u_3, \ldots, u_{\ell-1}.$
Furthermore, suppose that $(u_1 -u_\ell) \mod |n-m| \not= 0.$
This implies that there must exist $i$ ($1\le i < \ell$) such that
$(u_i -u_{i+1}) \mod |n-m| \not= 0.$
But there cannot be an edge between these two vertices,
and thus, the claimed path cannot exist.
This proves the lemma.
\end{proof}

We now proceed to the proof of Theorem~\ref{thm:networkdecomposition}.
We introduce the following two terms to have more compact arguments.
The term {\em direct edges}  refers to the edges connecting a transmitter with its corresponding receiver.
The term {\em cross edges}  refers to the edges connecting a transmitter with all receivers other than its corresponding receiver. Owing to the symmetry of the considered models, cross edges are the same for each receiver.

For Part {\em (1)}, consider the $(km, kn)$ model. The proof uses graph coloring with $k$ colors, identified by integers $\{ 0, 1, \ldots, k-1 \}.$
A vertex with label $u$ (where $u = 0, 1, 2,  \ldots, k \max (m,n)-1$) receives color $u \mod k$.
Hence, if two nodes (with labels $u$ and $v$) have different colors, this means that $(u-v) \mod k \not= 0.$
But this also implies that $(u-v) \mod k | m-n| \not= 0.$
Then, Lemma~\ref{Lemma-graphcoloring-nopath} implies that there cannot be a path between these vertices.
Hence, each color represents an independent graph.
Since we used $k$ colors, we thus have exactly $k$ independent graphs.
By the symmetry of the construction, all vertices and edges must be equally distributed,
and hence, each graph represents precisely an $(m, n)$ model. Alternatively, we now offer an explicit proof of this fact for the case $m<n.$ The case $m>n$ follows along the same lines. Let us first observe that in the $(km, kn)$ model (with $m<n$), there are direct edges between any transmitter vertex with label $u$ and its corresponding receiver vertex with the same label $u.$ Moreover, there are cross edges from transmitter vertex with label $u$ to receiver vertex with label $v$
if and only if $v-u = k(n-m).$ Now, consider the subgraph whose vertices are colored with the color identified by the integer $0.$ This subgraph thus contains all vertices with labels $0, k, 2k, \cdots.$
That is, in this case, we can express all vertex labels as $\ell k.$
Clearly, there are direct edges between transmitter vertex with label $\ell k$ and receiver vertex with the same label $\ell k.$
Moreover, there is a cross edge between transmitter vertex with label $\ell_1k$ and receiver vertex with label $\ell_2k$
if and only if $\ell_2k - \ell_1k = k(n-m),$ which can be rewritten as $\ell_2-\ell_1 = n-m.$
In other words, in our subgraph, if we assign new labels to each vertex by dividing the existing label by $k,$
meaning that a vertex previously labeled as $\ell k$ is now simply labeled as $\ell,$
we observe that each transmitter and each receiver has exactly $n$ vertices,
that there are direct edges between vertices with the same labels, and that there is a cross edge
from transmitter vertex with label $\ell_1$ to receiver vertex with label $\ell_2$ if and only if $\ell_2-\ell_1=n-m.$
But this is precisely an $(m,n)$ model (with $m<n$).
The same argument can be applied to the remaining $k-1$ colors.

For Part {\em (2)}, we consider the model $(2m+1, 2n+1).$
Again, we proceed by graph coloring, this time with two colors:
All vertices whose label $u$ (where $u = 0, 1, 2,  \ldots, 2\max (m,n)$)
is an even number receive one color, and all vertices whose label $u$ is an odd number
receive the other color.
Hence, if two nodes (with labels $u$ and $v$) have different colors, this means that $(u-v) \mod 2 \not= 0.$
But this also implies that $(u-v) \mod 2|n-m| \not= 0.$
Then, Lemma~\ref{Lemma-graphcoloring-nopath} implies that there cannot be a path between these vertices.
Hence, each of the two colors represents an independent graph.
The remainder of the proof is to establish the structure of these two subgraphs.
We will do so for the case $m<n;$ the case $m>n$ follows along the same lines.
First, we observe that in the $(2m+1, 2n+1)$ model (with $m<n$),
there is a direct edge between transmitter vertex with label $u$ and receiver vertex with the same label $u.$
Moreover, there is a cross edges from transmitter vertex with label $u$ to receiver vertex with label $v$
if and only if $v-u=2n+1-(2m+1)=2(n-m).$
Now, consider the subgraph corresponding to the odd-labeled vertices.
That is, in this case, we can express each label in the form $2\ell+1.$
We start by observing that each transmitter and each receiver has exactly $n$ odd-labeled vertices.
Moreover, there are direct edges between transmitter vertex with label $2\ell+1$
and receiver vertex with the same label $2\ell+1.$
Additionally, there is a cross edge from transmitter vertex with label $2\ell_1+1$ to receiver vertex with label $2\ell_2+1$
if and only if $2\ell_2+1-(2\ell_1+1)=2(n-m),$ which can be rewritten as $\ell_2-\ell_1=n-m.$
Relabeling the vertices of every transmitter and every receiver by subtracting 1 from the existing label
and then dividing by two (such that the vertex previously labeled as $2\ell+1$ is now labeled as $\ell$),
each transmitter and each receiver has vertices labeled $0, 1, \ldots, n-1,$
there are direct edges, and there is a cross edge between transmitter vertex $u$ and receiver vertex $v$
if and only if $v-u=n-m.$ Hence, this is an $(m,n)$ model (with $m<n$).
Next, consider the subgraph corresponding to the even-labeled vertices.
That is, in this case, we can express each label in the form $2\ell.$
We start by observing that each transmitter and each receiver has exactly $n+1$ even-labeled vertices.
Moreover, there are direct edges between transmitter vertex with label $2\ell$
and receiver vertex with the same label $2\ell.$
Additionally, there is a cross edge from transmitter vertex with label $2\ell_1$ to receiver vertex with label $2\ell_2$
if and only if $2\ell_2 - 2\ell_1=2(n-m),$ which can be rewritten as $\ell_2-\ell_1=n-m.$
Relabeling the vertices of every transmitter and every receiver by dividing by two (such that the vertex previously labeled as $2\ell$ is now labeled as $\ell$),
each transmitter and each receiver has vertices labeled $0, 1, \ldots, n,$
there are direct edges, and there is a cross edge between transmitter vertex $u$ and receiver vertex $v$
if and only if $v-u=n-m.$ Hence, this is an $(m+1,n+1)$ model (with $m<n$).

For Part {\em (3)}, we use graph coloring with $|n-m|$ colors,
identified by integers $\{ 0, 1, \ldots, |n-m|-1 \}.$
A vertex with label $u$ (where $u = 0, 1, 2,  \ldots, \max (m,n)-1$) receives color $u \mod |n-m|$.
Hence, if two nodes (with labels $u$ and $v$) have different colors, this means that $(u-v) \mod |n-m| \not= 0.$
Then, Lemma~\ref{Lemma-graphcoloring-nopath} implies that there cannot be a path between these vertices.
Hence, each color represents an independent graph.
But since there are $|n-m|$ colors, this means that there are $|n-m|$ independent subgraphs.

The remainder of the proof is to establish the structure of these subgraphs.
For convenience, we define $r  =   \left\lfloor  \frac{\min\{ m,n \}}{|n-m|} \right\rfloor$
and $a = \min\{ m,n \} \mod |n-m|.$

We explicitly consider the case $m<n.$
Note that in this case, there is a direct edge from transmitter vertex with label $u$ to receiver vertex with the same label $u.$
Moreover, there is a cross edge from transmitter vertex with label $u$ to receiver vertex with label $v$
if and only if $v-u=n-m.$
Also note that in this case, we have $a = m \mod (n-m) = n \mod (n-m),$ since for all integers $n$ and $m,$
we have $m \mod (n-m) = n \mod (n-m).$

First, let us consider the color $0.$ In the corresponding subgraph, each transmitter and each receiver
contains all vertices with labels $0, n-m, 2(n-m), \cdots.$
If $a \ge 1,$ that is, $m \mod (n-m) \ge 1$ or equivalently, $n \mod (n-m) \ge 1,$ then there are $\left \lceil  \frac{n}{n-m} \right \rceil = \left \lceil  1 + \frac{m}{n-m} \right \rceil = r+2$ such vertices.
Next, consider the edges in this subgraph.
Clearly, there is a direct edge from transmitter vertex with label $\ell(n-m)$ to receiver vertex with the same label $\ell(n-m).$
Moreover, there is a cross edge from transmitter vertex with label $\ell_1(n-m)$ to receiver vertex with label $\ell_2(n-m)$ if and only if $\ell_2(n-m)-\ell_1(n-m)=n-m,$
which can be rewritten as $\ell_2-\ell_1=1.$ Finally, assigning new vertex labels in the subgraph by dividing the existing label by $(n-m)$
reveals that the subgraph is exactly an $(r+1, r+2)$ model.
Otherwise, if $a=0,$ that is, $m \mod (n-m) =0$ or equivalently, $n \mod (n-m) =0$
(implying that both $m$ and $n$ are divisible by $(n-m)$), then there are only
$  \frac{n}{n-m}=  1 + \frac{m}{n-m} = r+1$ vertices in each transmitter and receiver.
Again, clearly, there are direct edges between each transmitter-receiver pair.
Moreover, there is a cross edge from transmitter vertex with label $\ell_1(n-m)$ to receiver vertex with label $\ell_2(n-m)$ if and only if $\ell_2(n-m)-\ell_1(n-m)=n-m,$
which can be rewritten as $\ell_2-\ell_1=1.$ Finally, assigning new vertex labels in the subgraph by dividing the existing labels by $(n-m)$
reveals that the subgraph is exactly an $(r, r+1)$ model.

Now, let us proceed to color $j,$ where $j \in \{1, 2, 3, n-m-1 \}.$ In the corresponding subgraph, each transmitter and each receiver
contains all vertices with labels $j, n-m+j, 2(n-m)+j, \cdots.$ If $a \ge j+1,$ that is, $m \mod (n-m) \ge j+1$ or equivalently, $n \mod (n-m) \ge j+1$
then there are $\left \lceil  \frac{n}{n-m} \right \rceil = \left \lceil  1 + \frac{m}{n-m} \right \rceil = r+2$ such vertices.
Next, consider the edges in this subgraph.
Clearly, there is a direct edge from transmitter vertex with label $\ell(n-m)+j$ to receiver vertex with the same label $\ell(n-m)+j.$
Moreover, there is a cross edge from transmitter vertex with label $\ell_1(n-m)+j$ to receiver vertex with label $\ell_2(n-m)+j$ if and only if $\ell_2(n-m)+j-(\ell_1(n-m)+j)=n-m,$
which can be rewritten as $\ell_2-\ell_1=1.$ Finally, assigning new vertex labels in the subgraph by subtracting $j$ from the existing label and dividing by $(n-m)$
reveals that the subgraph is exactly an $(r+1, r+2)$ model.
Otherwise, if $a \le j,$ that is, $m \mod (n-m) \le j$ or equivalently, $n \mod (n-m) \le j,$
then there are only $\left \lfloor  \frac{n}{n-m} \right \rfloor = \left \lfloor  1 + \frac{m}{n-m} \right \rfloor = r+1$
vertices in each transmitter and receiver.
Again, clearly, there are direct edges between each transmitter-receiver pair.
Moreover, there is a cross edge from transmitter vertex with label $\ell_1(n-m)+j$ to receiver vertex with label $\ell_2(n-m)+j$ if and only if $\ell_2(n-m)+j-(\ell_1(n-m)+j)=n-m,$
or $\ell_2-\ell_1=1.$ Finally, assigning new vertex labels in the subgraph by subtracting $j$ from the existing label and then dividing by $(n-m)$
reveals that the subgraph is exactly an $(r, r+1)$ model.

In summary, we observe that color $j$ (with $j \in \{0, 1, 2, 3, n-m-1 \}$) represents an $(r+1, r+2)$ model if $aÊ\ge j+1$
and an $(r, r+1)$ model if $a < j+1.$ In other words, the colors with indices $j=0, 1, \cdots, a-1$ represent $(r+1, r+2)$ models,
and colors with indices $j=a, a+1, \cdots, (n-m)$ represent $(r, r+1)$ models.
This means that in the decomposition,
we will have exactly $a$ models of type $(r+1, r+2)$ and the rest, i.e., $n-m - a$ models
of type $(r, r+1).$

The case $m>n$ follows along the same lines and is omitted for brevity.
It is found in this case that we obtain $a$ models of the type $(r+2, r+1)$
and $|n-m|-a$ models of the type $(r+1, r).$

\bibliographystyle{ieeetr}
\bibliography{bib_function}

\begin{IEEEbiographynophoto}{Changho Suh}
(S'10-M'12) is an Ewon Assistant Professor in the School of Electrical Engineering at Korea Advanced Institute of Science and Technology (KAIST) since 2012. He received the B.S. and M.S. degrees in Electrical Engineering from KAIST in 2000 and 2002 respectively, and the Ph.D. degree in Electrical Engineering and Computer Sciences from UC-Berkeley in 2011. From 2011 to 2012, he was a postdoctoral associate at the Research Laboratory of Electronics in MIT. From 2002 to 2006, he had been with the Telecommunication R\&D Center, Samsung Electronics.

Dr. Suh received the 2015 Hadong Young Engineer Award from the Institute of Electronics and Information Engineers, the 2013 Stephen O. Rice Prize from the IEEE Communications Society, the David J. Sakrison Memorial Prize from the UC-Berkeley EECS Department in 2011, and the Best Student Paper Award of the IEEE International Symposium on Information Theory in 2009.
\end{IEEEbiographynophoto}

\begin{IEEEbiographynophoto}{Naveen Goela}
received undergraduate degrees in computer science and mathematics, and a M.\,Eng. degree in computer science in 2004, all from the Massachusetts Institute of Technology (MIT) in Cambridge, MA, USA. His research at MIT focused on multi-camera visual hull reconstruction, and motion graphs. In 2013, he received the Ph.D. in electrical engineering and computer science with designated emphasis in statistics from the University of California, Berkeley. His past experience includes a research internship at Mitsubishi Electric Research Labs (MERL) in 2006, a visiting scholarship at Ecole Polytechnique F\'ed\'erale de Lausanne (EPFL) in Switzerland during 2011-2012, and a post-doctoral research position at Qualcomm Research in Berkeley, CA, USA during 2013-2014. Currently, he is a member of Technicolor Research in Los Altos, CA, USA. From 2007-2010, he was awarded the U.S. National Defense Science and Engineering Graduate (NDSEG) Fellowship. He received Qualcomm's company-wide ImpaQt award in 2014 for patents related to 5G wireless communications, and Technicolor's company-wide innovation awards in 2015 for research focusing on the signal processing of light signals. His research interests include networks, statistics, machine learning, and information theory.
\end{IEEEbiographynophoto}

\begin{IEEEbiographynophoto}{Michael Gastpar}
received the Dipl. El.-Ing. degree from the Eidgen\"{o}ssishe Technische Hochschule (ETH), Z\"urich, Switzerland, in 1997, the M.S. degree in electrical engineering from the University of Illinois at Urbana-Champaign, Urbana, IL, USA, in 1999, and the Doctorat \`es Science degree from the Ecole Polytechnique F\'ed\'erale (EPFL), Lausanne, Switzerland, in 2002. He was also a student in engineering and philosophy at the Universities of Edinburgh and Lausanne.

During the years 2003-2011, he was an Assistant and tenured Associate Professor with the Department of Electrical Engineering and Computer Sciences at the University of California, Berkeley. Since 2011, he has been a Professor in the School of Computer and Communication Sciences, Ecole Polytechnique F\'ed\'erale (EPFL), Lausanne, Switzerland. He is also a professor at Delft University of Technology, The Netherlands. He was a Researcher with the Mathematics of Communications Department,
Bell Labs, Lucent Technologies, Murray Hill, NJ. His research interests are in network information theory and related coding and signal processing techniques, with applications to sensor networks and neuroscience.

Dr. Gastpar received the IEEE Communications Society and Information Theory Society Joint Paper Award in 2013 and the EPFL Best Thesis Award in 2002. He was an Information Theory Society Distinguished Lecturer (2009-2011), an Associate Editor for Shannon Theory for the IEEE TRANSACTIONS ON INFORMATION THEORY (2008-2011), and he has served as Technical Program Committee Co-Chair for the 2010 International Symposium on Information Theory, Austin, TX.
\end{IEEEbiographynophoto}

\end{document}